\documentclass[11pt]{article}

\usepackage[utf8]{inputenc}

\usepackage{amsmath,amsfonts,amsthm,amssymb,color}
\usepackage[usenames,dvipsnames,svgnames,table]{xcolor}
\definecolor{darkgreen}{rgb}{0.0,0,0.9}
\usepackage{mathtools}
\usepackage{authblk}
\usepackage{fullpage}
\usepackage{parskip}
\usepackage{comment}
\usepackage{tikz}
\usepackage{bbm}
\usepackage{dsfont}
\usepackage[sc]{mathpazo}
\usepackage[basic]{complexity}
\usepackage{algorithm2e}
\usepackage[colorlinks=true,
citecolor=OliveGreen,linkcolor=BrickRed,urlcolor=BrickRed,pdfstartview=FitH]{hyperref}
\usepackage[capitalize,nameinlink]{cleveref}
\usepackage{tcolorbox}
\usepackage[short]{optidef}

\newtcolorbox{wbox}
{
	colback  = white,
}

\SetKwInOut{Input}{Input}
\SetKwInOut{Output}{Output}
\SetKwFunction{Uncross}{\textsc{Uncross}}
\SetKwFunction{MergeUncross}{\textsc{MergeUncross}}
\SetKwFunction{PerfectMatching}{\textsc{PerfectMatching}}
\SetKwBlock{InParallel}{in parallel do}{end}
\SetKwFor{ParallelFor}{for}{in parallel do}{end}

\newcommand*{\suppress}[1]{}

\newcommand*{\cR}{\mathcal{R}}

\makeatletter
\def\thm@space@setup{%
	\thm@preskip= 10pt
	\thm@postskip=\thm@preskip 
}
\makeatother

\makeatletter
\renewcommand{\paragraph}{%
	\@startsection{paragraph}{4}%
	{\z@}{5pt}{-1em}%
	{\normalfont\normalsize\bfseries}%
}
\makeatother

\newtheorem{theorem}{Theorem}
\newtheorem{lemma}[theorem]{Lemma}

\newtheorem{proposition}[theorem]{Proposition}

\theoremstyle{definition}
\newtheorem{definition}[theorem]{Definition}

\newtheorem{remark}[theorem]{Remark}

\newtheorem{alg}[theorem]{Mechanism}

\newtheorem{example}[theorem]{Example}

\newenvironment{fminipage}%
{\begin{Sbox}\begin{minipage}}%
		{\end{minipage}\end{Sbox}\fbox{\TheSbox}}

\newcommand{\opt}{\mbox{\rm OPT}}

\title{The General Graph Matching Game: \\
 Approximate Core}

\author[1]{Vijay V.~Vazirani\footnote{Supported in part by NSF grant CCF-1815901.}}

\affil[1]{University of California, Irvine}

\date{}

\begin{document}
	\maketitle
	
	\begin{abstract}
The classic paper of Shapley and Shubik \cite{Shapley1971assignment} characterized the core of the assignment game using ideas from matching theory and LP-duality theory and their highly non-trivial interplay. Whereas the core of this game is always non-empty, that of the general graph matching game can be empty.

This paper salvages the situation by giving an imputation in the $2/3$-approximate core for the latter. This bound is best possible, since it is the integrality gap of the natural underlying LP. Our profit allocation method goes further: the multiplier on the profit of an agent is often better than ${2 \over 3}$ and lies in the interval $[{2 \over 3}, 1]$, depending on how severely constrained the agent is.   

Next, we provide new insights showing how discerning core imputations of an assignment games are by studying them via the lens of complementary slackness. We present a relationship between the competitiveness of individuals and teams of agents and the amount of profit they accrue in imputations that lie in the core, where by {\em competitiveness} we mean whether an individual or a team is matched in every/some/no maximum matching. This also sheds light on the phenomenon of degeneracy in assignment games, i.e., when the maximum weight matching is not unique. 

The core is a quintessential solution concept in cooperative game theory. It contains all ways of distributing the total worth of a game among agents in such a way that no sub-coalition has incentive to secede from the grand coalition. Our imputation, in the $2/3$-approximate core, implies that a sub-coalition will gain at most a $3/2$ factor by seceding, and less in typical cases.
\end{abstract}

\bigskip
\bigskip
\bigskip
\bigskip
\bigskip
\bigskip
\bigskip
\bigskip
\bigskip
\bigskip
\bigskip
\bigskip
\bigskip
\bigskip
\bigskip
\bigskip
\bigskip
\bigskip

\pagebreak

\section{Introduction}
\label{sec.intro}

The matching game forms one of the cornerstones of cooperative game theory. This game can also be viewed as a matching market in which utilities of the agents are stated in monetary terms and side payments are allowed, i.e., it is a {\em transferable utility (TU) market}. A key solution concept in this theory is that of the {\em core}, which captures all possible ways of distributing the total worth of a game among individual agents in such a way that the grand coalition remains intact, i.e., a sub-coalition will not be able to generate more profits by itself and therefore has no incentive to secede from the grand coalition. For an extensive coverage of these notions, see the book by Moulin \cite{Moulin2014cooperative}.    

When restricted to bipartite graphs, the matching game is called the {\em assignment game}. The classic paper of Shapley and Shubik \cite{Shapley1971assignment} characterized profit-sharing methods that lie in the core of such games by using ideas from matching theory and LP-duality theory and their highly non-trivial interplay; in particular, the core is always non-empty.

On the other hand, for games defined over general graphs, the core is not guaranteed to be non-empty, see Section for an easy proof. The purpose of this paper is to salvage the situation to the extent possible by giving a notion of approximate core for such games. The approximation factor we achieve is $2/3$. This is best possible, since it is the integrality gap of the underlying LP; this follows easily from an old result of Balinski \cite{Balinski1965integer} characterizing the vertices of the polytope defined by the constraints of this LP. It turns out that an optimal integral solution to this LP gives the worth of the game and the constraints of dual of this LP must be respected by any profit-sharing mechanism.  

An interesting feature of our profit-sharing mechanism is that it restricts only the most severely constrained agents to a multiplier of $2/3$, and the less severely constrained an agent is, the better is her multiplier, all the way to 1; bipartite graphs belong to the last category. One way of stating the improved factor is: if the underlying graph has no odd cycles of length less than $2k+1$, then our factor is ${{2k} \over {2k+1}}$.

The following setting, taken from \cite{Eriksson2001stable} and \cite{Biro2012computing}, vividly captures the underlying issues. Suppose a tennis club has a set $V$ of players who can play in an upcoming doubles tournament. Let $G = (V, E)$ be a graph whose vertices are the players and an edge $(i, j)$ represents the fact that players $i$ and $j$ are compatible doubles partners. Let $w$ be an edge-weight function for $G$, where $w_{i j}$ represents the expected earnings if $i$ and $j$ do partner in the tournament. Then the total worth of agents in $V$ is the weight of a maximum weight matching in $G$. Assume that the club picks such a matching $M$ for the tournament. The question is how to distribute the total profit among the agents --- strong players, weak players and unmatched players --- so that no subset of players feel they will be better off seceding and forming their own tennis club. 

In Section \ref{sec.Degen}, we provide new insights on imputations in the core of an assignment game by studying them through the lens of complementary slackness. We present a relationship between the competitiveness of individuals and teams of agents and the amount of profit they accrue,  where by {\em competitiveness} we mean whether an individual or a team is matched in every/some/no maximum matching. Theorem \ref{thm.degen} shows once again how discerning core imputations are. 

Additionally, Theorem \ref{thm.degen} sheds light on the phenomenon of degeneracy in assignment games, i.e., when the maximum weight matching is not unique. Shapley and Shubik had mentioned this phenomenon; however, they claimed that ``in the most common case'' the optimal assignment will be unique. Furthermore,  their suggestion for  dealing with degeneracy was to perturb the edge weights of $G$ to make the optimal assignment unique, so that they only needed to address that case. However, perturbing the weights destroys crucial information contained in the original instance. 

Over the years, this phenomenon has been studied by other researchers. In an interesting work, Nunez and Rafels \cite{Nunez-Dimension}, studied relationships between degeneracy and the dimension of the core. The defined an agent to be {\em active} if her profit is not constant across the various imputations in the core, and non-active otherwise. Clearly, this notion has much to do with the dimension of the core, e.g., it is easy to see that if all agents are non-active, the core must be zero-dimensional. They prove that if all agents are active, then the core is full dimensional if and only if the game is non-degenerate. Furthermore, if there are exactly two optimal matchings, then the core can have any dimension between 1 and $m-1$, where $m$ is the smaller of $|U|$ and $|V|$; clearly, $m$ is an upper bound on the dimension.

In another interesting work, Chambers and Echenique \cite{Chambers2015core} study the following question: Given the entire set of optimal matchings of a game on $m = |U|$, $n = |V|$ agents, is there an $m \times n$ surplus matrix which has this set of optimal matchings. They give necessary and sufficient conditions for the existence of such a matrix.  

 Finally we remark that all facts used in this paper were available in 1971, the year of the Shapley-Shubik paper. Therefore, all our results could have been obtained then or at any point in the intervening half a century.

\section{Definitions and Preliminary Facts}
\label{sec.matching-game}

\begin{definition}
	\label{def.matching-game}
	The {\em general graph matching game} consists of an undirected graph $G = (V, E)$ and an edge-weight function $w$. The vertices $i \in V$ are the agents and an edge $(i, j)$ represents the fact that agents $i$ and $j$ are eligible for an activity, for concreteness, let us say that they are eligible to participate as a doubles team in a tournament. If $(i, j) \in E$, $w_{i j}$ represents the profit generated if $i$ and $j$ play in the tournament. The {\em worth} of a coalition $S \subseteq V$ is defined to be the maximum profit that can be generated by teams within $S$ and is denoted by $p(S)$. Formally, $p(S)$ is defined to be the weight of a maximum weight matching in the graph $G$ restricted to vertices in $S$ only. The {\em characteristic function} of the matching game is defined to be $p: 2^{V} \rightarrow \cR_+$.   
\end{definition}     

Among the possible coalitions, the most important one is of course $V$, the {\em grand coalition}. 

\begin{definition}
	\label{def.core}
	An {\em imputation} gives a way of dividing the worth of the game, $p(V)$, among the agents. Formally, it is a function $v: {V} \rightarrow \cR_+$ such that $\sum_{i \in V} {v(i)} = p(V)$. An imputation $t$ is said to be in the {\em core of the matching game} if for any coalition $S \subseteq V$, there is no way of dividing $p(S)$ among the agents in $S$ in such a way that all agents are at least as well off and at least one agent is strictly better off than in the imputation $t$.
\end{definition}
 
We next describe the characterization of the core of the assignment game given by Shapley and Shubik \cite{Shapley1971assignment}. Shapley and Shubik had described this game in the context of the housing market in which agents are of two types, buyers and sellers. They had shown that each imputation in the core of this game gives rise to unique prices for all the houses. In this paper we will present the assignment game in a variant of the tennis setting given in the Introduction; this will obviate the need to define ``prices'', hence leading to simplicity. 

Suppose a coed tennis club has sets $U$ and $V$ of women and men players, respectively, who can participate in an upcoming mixed doubles tournament. Assume $|U| = m$ and $|V| = n$, where $m, n$ are arbitrary. Let $G = (U, V, E)$ be a bipartite graph whose vertices are the women and men players and an edge $(i, j)$ represents the fact that agents $i \in U$ and $j \in V$ are  eligible to participate as a mixed doubles team in the tournament. Let $w$ be an edge-weight function for $G$, where $w_{i j}$ represents the expected earnings if $i$ and $j$ do participate as a team in the tournament. We will assume that $w(i, j) > 0$; this is reasonable since $i$ and $j$ are compatible. Once again, the total worth of the game is the weight of a maximum weight matching in $G$.

The linear program (\ref{eq.core-primal}) gives the LP-relaxation of the problem of finding such a matching. In this program, variable $x_{ij}$ indicates the extent to which edge $(i, j)$ is picked in the solution. Matching theory tells us that this LP always has an integral optimal solution; the latter is a maximum weight matching in $G$.

	\begin{maxi}
		{} {\sum_{(i, j) \in E}  {w_{ij} x_{ij}}}
			{\label{eq.core-primal-bipartite}}
		{}
		\addConstraint{\sum_{(i, j) \in E} {x_{ij}}}{\leq 1 \quad}{\forall i \in U}
		\addConstraint{\sum_{(i, j) \in E} {x_{ij}}}{\leq 1 }{\forall j \in V}
		\addConstraint{x_{ij}}{\geq 0}{\forall (i, j) \in E}
	\end{maxi}

Taking $u_i$ and $v_j$ to be the dual variables for the first and second constraints of (\ref{eq.core-primal-bipartite}), we obtain the dual LP: 

 	\begin{mini}
		{} {\sum_{i \in U}  {u_{i}} + \sum_{j \in V} {v_j}} 
			{\label{eq.core-dual-bipartite}}
		{}
		\addConstraint{ u_i + v_j}{ \geq w_{ij} \quad }{\forall (i, j) \in E}
		\addConstraint{u_{i}}{\geq 0}{\forall i \in U}
		\addConstraint{v_{j}}{\geq 0}{\forall j \in V}
	\end{mini}

The definition of an imputation, Definition \ref{def.core}, needs to be modified in an obvious way to distinguish the profit shares of women and men players. We will denote an imputation by $(u, v)$. 

\begin{theorem}
	\label{thm.SS}
	(Shapley and Shubik \cite{Shapley1971assignment})
	The imputation $(u, v)$ is in the core of the assignment game if and only if it is an optimal solution to the dual LP, (\ref{eq.core-dual-bipartite}). 
\end{theorem}


For general graphs, the LP relaxation of the maximum weight matching problem is an enhancement of that for bipartite graphs via odd set constraints, as given below in (\ref{eq.core-primal-general}). The latter constraints are exponential in number.

	\begin{maxi}
		{} {\sum_{(i, j) \in E}  {w_{ij} x_{ij}}}
			{\label{eq.core-primal-general}}
		{}
		\addConstraint{\sum_{(i, j) \in E} {x_{ij}}}{\leq 1 \quad}{\forall i \in V}
		\addConstraint{\sum_{(i, j) \in S} {x_{ij}}}{\leq {{(|S|-1)} \over 2} \quad}{\forall S \subseteq V, \ S \ \mbox{odd}}
		\addConstraint{x_{ij}}{\geq 0}{\forall (i, j) \in E}
	\end{maxi}

The dual of this LP has, in addition to variables corresponding to vertices, $v_i$, exponentially many more variables corresponding to odd sets, $z_S$, as given in (\ref{eq.core-dual-general}). As a result, the entire worth of the game does not reside on vertices only --- it also resides on odd sets.

 	\begin{mini}
		{} {\sum_{i \in V}  {v_{i}} + \sum_{S \subseteq V, \ \mbox{odd}} {z_S}} 
			{\label{eq.core-dual-general}}
		{}
		\addConstraint{v_i + v_j + \sum_{S \, \ni \, i, j}{z_S}}{ \geq w_{ij} \quad }{\forall (i, j) \in E}
		\addConstraint{v_{i}}{\geq 0}{\forall i \in V}
		\addConstraint{z_{S}}{\geq 0}{\forall S \subseteq V, \ S \ \mbox{odd}}
	\end{mini}

There is no natural way of dividing $z_S$ among the vertices in $S$ to restore core properties. The situation is more serious than that: it turns out that in general, the core of a non-bipartite game may be empty. 

A (folklore) proof of the last fact goes as follows: Consider the graph $K_3$, i.e., a clique on three vertices, $i, j, k$, with a weight of 1 on each edge. Any maximum matching in $K_3$ has only one edge, and therefore the worth of this game is 1. Suppose there is an imputation $v$ which lies in the core. Consider all three two-agent coalitions. Then, we must have:
\[ v(i) + v(j) \geq 1, \ \ \ \  v(j) + v(k) \geq 1 \ \ \ \ \mbox{and} \ \ \ \ v(i) + v(k) \geq 1 .\]
This implies $v(i) + v(j) + v(k) \geq 3/2$ which exceeds the worth of the game, giving a contradiction.

Observe however, that if we distribute the worth of this game as follows, we get a $2/3$-approximate core allocation: $v(i) = v(j) = v(k) = 1/3$. Now each edge is covered to the extent of $2/3$ of its weight. In Section \ref{sec.App-Core} we show that such an approximate core allocation can always be obtained for the general graph matching game. 

\begin{remark}
	\label{rem.TU}
	In the assignment game, monetary transfers are required only between a buyer-seller pair who are involved in a trade. In contrast, observe that in the approximate core allocation given above for $K_3$, transfers are even made from agents who make trade to agents who don't, thereby exploiting the TU market's capabilities more completely. 
\end{remark}

\begin{definition}
	\label{def.app-core}
	Let $p: 2^{V} \rightarrow \cR_+$ be the characteristic function of a game and let $1 \geq \alpha > 0$. An imputation $t: V \rightarrow \cR_+$ is said to be in the {\em $\alpha$-approximate core} of the game if:
	\begin{enumerate}
		\item The total profit allocated by $t$ is at most the worth of the game, i.e., 
		 \[\sum_{i \in V} {t_i} \leq p(V) .\]
		\item  The total profit accrued by agents in a sub-coalition $S \subseteq V$ is at least $\alpha$ fraction of the profit which $S$ can generate by itself, i.e., 
		\[ \forall S \subseteq V: \ \sum_{i \in S} {t_i} \geq \alpha \cdot p(S) .\]
	\end{enumerate} 
\end{definition}

If imputation $t$ is in the $\alpha$-approximate core of a game, then the ratio of the total profit of any sub-coalition on seceding from the grand coalition to its profit while in the grand coalition is bounded by a factor of at most ${1 \over {\alpha}}$.

In Section \ref{sec.App-Core} we will need the following notion.

\begin{definition}
	\label{def.int-gap}
	Consider a linear programming relaxation for a maximization problem. For an instance $I$ of the problem, let $\opt(I)$ denote the weight of an optimal solution to $I$ and let $\opt_f(I)$ denote the weight of an optimal solution to the LP-relaxation of $I$. Then, the {\em integrality gap} of this LP-relaxation is defined to be:
	\[ \inf_I {{\opt(I)} \over {\opt_f(I)}} . \] 
\end{definition}

\section{A $2/3$-Approximate Core for the Matching Game}
\label{sec.App-Core}

We will work with the following LP-relaxation of the maximum weight matching problem, (\ref{eq.core-primal}). This relaxation always has an integral optimal solution in case $G$ is bipartite, but not in general graphs. In the latter, its optimal solution is a maximum weight fractional matching in $G$. 

	\begin{maxi}
		{} {\sum_{(i, j) \in E}  {w_{ij} x_{ij}}}
			{\label{eq.core-primal}}
		{}
		\addConstraint{\sum_{(i, j) \in E} {x_{ij}}}{\leq 1 \quad}{\forall i \in V}
		\addConstraint{x_{ij}}{\geq 0}{\forall (i, j) \in E}
	\end{maxi}

Taking $v_i$ to be dual variables for the first constraint of (\ref{eq.core-primal}), we obtain  LP (\ref{eq.core-dual}). Any feasible solution to this LP is called a {\em cover} of $G$ since for each edge $(i, j)$, $v_i$ and $v_j$ cover edge $(i, j)$ in the sense that $v_i + v_j \geq w_{ij}$. An optimal solution to this LP is a {\em minimum cover}. We will say that $v_i$ is the {\em profit} of vertex $i$.

 	\begin{mini}
		{} {\sum_{i \in V}  {v_{i}}} 
			{\label{eq.core-dual}}
		{}
		\addConstraint{v_i + v_j}{ \geq w_{ij} \quad }{\forall (i, j) \in E}
		\addConstraint{v_{i}}{\geq 0}{\forall i \in V}
	\end{mini}

By the LP Duality Theorem, the weight of a maximum weight fractional matching equals the total profit of a minimum cover. If for graph $G$, LP (\ref{eq.core-primal}) has an integral optimal solution, then it is easy to see that an optimal dual solution gives a way of allocating the total worth which lies in the core. Deng et. al. \cite{Deng1997algorithms} prove that the core of this game is non-empty if and only if LP (\ref{eq.core-primal}) has an integral optimal solution.
	
We will say that a solution $x$ to LP (\ref{eq.core-primal}) is {\em half-integral} if for each edge $(i, j)$, $x_{ij}$ is 0, 1/2 or 1. Balinski \cite{Balinski1965integer} showed that the vertices of the polytope defined by the constraints of LP (\ref{eq.core-primal}) are half-integral, see Theorem \ref{thm.Balinski}. 

\begin{theorem}
	\label{thm.Balinski}
	(Balinski \cite{Balinski1965integer})
	The vertices of the polytope defined by the constraints of LP (\ref{eq.core-primal}) are half-integral.
\end{theorem}

As a consequence of Theorem \ref{thm.Balinski}, LP (\ref{eq.core-primal}) always has a half-integral  optimal solution. Biro et. al \cite{Biro2012computing} gave an efficient algorithm for finding an optimal half-integral matching by using an idea of Nemhauser and Trotter  \cite{Nemhauser1975vertex} of doubling edges, hence obtaining an efficient algorithm for determining if the core of the game is non-empty. 

Our mechanism starts by using the doubling idea of \cite{Nemhauser1975vertex}. Transform $G = (V, E)$ with edge-weights $w$ to graph $G' = (V', E')$ and edge weights $w'$ as follows. Corresponding to each $i \in V$, $V'$ has vertices $i'$ and $i''$, and corresponding to each edge $(i, j) \in E$, $E'$ has edges $(i', j'')$ and $(i'', j')$ each having a weight of $w_{ij}/2$. 

Since each cycle of length $k$ in $G$ is transformed to a cycle of length $2k$ in $G'$, the latter graph has only even length cycles and is bipartite. A maximum weight matching and a minimum cover for $G'$ can be computed in polynomial time \cite{LP.book}, say $x'$ and $v'$, respectively. Next, let
\[ x_{ij} \ = \ {1 \over 2} \cdot (x_{i', j''} + x_{i'', j'}) \ \ \ \ \mbox{and} \ \ \ \  v_i = (v_{i'} + v_{i''}) .\]
It is easy to see that the weight of $x$ equals the value of $v$, thereby implying that $v$ is an optimal cover. 

\begin{lemma}
	\label{lem.opt-primal-dual}
	$x$ is a maximum weight half-integral matching and $v$ is an optimal cover in $G$. 
\end{lemma}

\begin{proof}
We will first use the fact that $v'$ is a feasible cover for $G'$ to show that $v$ is a feasible cover for $G$. Corresponding to each edge $(i, j)$ in $G$, we have two edges in $G'$ satisfying:
\[ v'_{i'} + v'_{j''} \geq {1 \over 2} \cdot w_{ij} \ \ \ \  \mbox{and} \ \ \ \ v'_{i''} + v'_{j'} \geq {1 \over 2} \cdot w_{ij} .\]
Therefore, in $G$, $v_i + v_j \geq w_{ij}$, implying feasibility of $v$. 

	By the LP-duality theorem, the weight of $x'$ equals the value of $v'$ in $G'$. Corresponding to each edge $(i, j)$ in $G$, we have:
\[ x_{ij} \cdot w_{ij} \ = \ \left( {1 \over 2} \cdot (x_{i', j''} + x_{i'', j'}) \right) \cdot w_{ij} \ = \ x_{i', j''} \cdot {{w'_{ij}}} + x_{i'', j'} \cdot {{w'_{ij}}} . \]
Adding over all edges, we get that the weight of $x$ in $G$ equals the weight of $x'$ in $G'$. Furthermore, the profit of $i$ equals the sum of profits of ${i'}$ and ${i''}$. Therefore the value of $v$ in $G$ equals the value of $v'$ in $G'$.  

Putting it together, we get that the weight of $w$ equals the value of $v$, implying optimality of both. Clearly, $x$ is half-integral. The lemma follows.
\end{proof}

Edges that are set to half in $x$ form connected components which are either paths or cycles. For any such path, consider the two matchings obtained by picking alternate edges. The half-integral solution for this path is a convex combination of these two integral matchings. Therefore both these matchings must be of equal weight, since otherwise we can obtain a heavier matching. Pick any of them. Similarly, if a cycle is of even length, pick alternate edges and match them. This transforms $x$ to a maximum weight half-integral matching in which all edges that are set to half form disjoint odd cycles. Henceforth we will assume that $x$ satisfies this property.  

	Let $C$ be a half-integral odd cycle in $x$ of length $2k+1$, with consecutive vertices $i_1, \ldots i_{2k+1}$. Let $w_C = w_{i_1, i_2} + w_{i_2, i_3} + \ldots + w_{i_{2k+1}, i_1}$ and $v_C = v_{i_1} + \ldots v_{i_{2k+1}}$. On removing any one vertex, say $i_j$, with its two edges from $C$, we are left with a path of length $2k-1$. Let $M_j$ be the matching consisting of the $k$ alternate edges of this path and let $w(M_j)$ be the weight of this matching.

\begin{lemma}
	\label{lem.odd-cycle}
	Odd cycle $C$ satisfies:
	\begin{enumerate}
		\item $w_C = 2 \cdot v_C$
		\item $C$ has a unique cover: $v_{i_j} = v_C - w(M_j)$, for $1 \leq j \leq 2k + 1$.
	\end{enumerate}
\end{lemma}   

\begin{proof}
	{\bf 1).}  We will use the fact that $x$ and $v$ are optimal solutions to LPs (\ref{eq.core-primal}) and (\ref{eq.core-dual}), respectively. By the primal complementary slackness condition, for $1 \leq j \leq 2k + 1$, $w_{i_j, i_{j+1}} = v_{i_j} + v_{i_{j+1}}$, where addition in the subindices is done modulo $2k+1$; this follows from the fact that $x_{i_j, i_{j+1}} > 0$. Adding over all vertices of $C$ we get $w_C = 2 \cdot v_C$. 
	
	{\bf 2).}  By the equalities established in the proof of the first part, we get that for $1 \leq j \leq 2k + 1$, $v_C =  v_{i_j} + w(M_j)$. Rearranging terms gives the lemma. 
\end{proof}

Let $M'$ be heaviest matching among $M_j$, for $1 \leq j \leq 2k + 1$.

\begin{lemma}
	\label{lem.heaviest}
	\[ w(M') \geq {{2k} \over {2k+1}} \cdot v_C \]
\end{lemma}

\begin{proof}
	Adding the equality established in the second part of Lemma \ref{lem.odd-cycle} for all $2k+1$ values of $j$ we get:
	\[ \sum_{j = 1}^{2k+1} {w(M_j)} \ = \ (2k) \cdot v_C \]
	Since $M'$ is the heaviest of the $2k+1$ matchings in the summation, the lemma follows. 
\end{proof}
	
Modify the half-integral matching $x$ to obtain an integral matching $T$ in $G$ as follows. First pick all edges $(i, j)$ such that $x_{ij} = 1$ in $T$. Next, for each odd cycle $C$, find the heaviest matching $M'$ as described above and pick all its edges. 

\begin{definition}
	\label{def.app-cover}
Let $1 > \alpha > 0$. A function $c: V \rightarrow \cR_+$ is said to be an {\em ${\alpha}$-approximate cover} for $G$ if 
\[ \forall (i, j) \in E: \ \  c_i + c_j \geq \alpha \cdot w_{ij} \]
\end{definition}

	Define function $f: V \rightarrow [{2 \over 3}, 1]$ as follows: $\forall i \in V$:
			\begin{equation*}
    f(i) =
    \begin{cases*}
      {{2k} \over {2k+1}} & if $i$ is in a half-integral cycle of length $2k+1$.  \\
      1        & if $i$ is not in a half-integral cycle.
    \end{cases*}
  \end{equation*}

	Next, modify cover $v$ to obtain an approximate cover $c$ as follows: $\forall i \in V: \ c_i = f(i) \cdot v_i$.

\begin{lemma}
	\label{lem.app-cover}
$c$ is a ${2 \over 3}$-approximate cover for $G$. 
\end{lemma} 

\begin{proof}
	Consider edge $(i, j) \in E$. Then
	\[ c_i + c_j \ = \ f(i) \cdot v_i + f(j)\ \cdot	 v_j \ \geq {2 \over 3} \cdot (v_i + v_j) \ \geq \ {2 \over 3} \cdot w_{ij} , \]
	where the first inequality follows from the fact that $\forall i \in V, \ f(i) \geq {2 \over 3}$ and the second follows from the fact that $v$ is a cover for $G$.
\end{proof}

The mechanism for obtaining imputation $c$ is summarized as Mechanism \ref{alg.core}.

\bigskip

\setcounter{figure}{1} 

\begin{figure}

	\begin{wbox}
		\begin{alg}
		\label{alg.core}
		{\bf (${2/3}$-Approximate Core Imputation)}\\
		\\
		\begin{enumerate}
			\item Compute $x$ and $v$, optimal solutions to LPs (\ref{eq.core-primal}) and (\ref{eq.core-dual}),  where $x$ is half-integral.
			\item Modify $x$ so all half-integral edges form odd cycles.
			\item $\forall i \in V$, compute:
			\begin{equation*}
    f(i) =
    \begin{cases*}
      {{2k} \over {2k+1}} & if $i$ is in a half-integral cycle of length $2k+1$.  \\
      1        & otherwise.
    \end{cases*}
  \end{equation*}
  \item $\forall i \in V$: \ \ $c_i \leftarrow f(i) \cdot v_i$.
		\end{enumerate} 
		\bigskip
		Output $c$.
		\end{alg}
	\end{wbox}
\end{figure}

\begin{theorem}
	\label{thm.main}
	The imputation $c$ is in the ${2 \over 3}$-approximate core of the general graph matching game.
\end{theorem}
	
\begin{proof}
We need to show that $c$ satisfies the two conditions given in Definition \ref{def.app-core}, for $\alpha = {2 \over 3}$.

{\bf 1).}	By Lemma \ref{lem.heaviest}, the weight of the matched edges picked in $T$ from a half-integral odd cycle $C$ of length $2k+1$ is $\geq f(k) \cdot v_C \ = \ \sum_{i \in C} {c(i)}$. Next remove all half-integral odd cycles from $G$ to obtain $G'$. Let $x'$ and $v'$ be the projections of $x$ and $v$ to $G'$. 

By the first part of Lemma \ref{lem.odd-cycle}, the total decrease in weight in going from $x$ to $x'$ equals the total decrease in value in going from $v$ to $v'$. Therefore, the weight of $x'$ equals the total value of $v'$. Finally, observe that in $G'$, $T$ picks an edge $(i, j)$ if and only if $x' _{ij} = 1$ and $\forall i \in G', \ c_i = v' _i$. 

Adding the weight of the matching and the value of the imputation $c$ over $G'$ and all half-integral odd cycles we get $w(T) \geq \sum_{i \in V} {c_i}$.

{\bf 2).} Consider a coalition $S \subseteq V$. Then $p(S)$ is the weight of a maximum weight matching in $G$ restricted to $S$. Assume this matching is $(i_1, j_1), \ldots (i_k, j_k)$, where $i_1, \ldots i_k$ and $j_1, \ldots j_k \in S$. Then $p(S) = (w_{i_1 j_1} + \ldots + w_{i_k j_k})$. By Lemma \ref{lem.app-cover}, 
\[ c_{i_l} + c_{j_l} \ \geq \ {2 \over 3} \cdot w_{i_l , j_l}, \ \ \mbox{for} \ 1 \leq l \leq k .\] 
Adding all $k$ terms we get: 
\[ \sum_{i \in S} {c_i} \ \geq \ {2 \over 3} \cdot p(S) .\]
\end{proof}

We next show that the factor of $2/3$ cannot be improved by giving a {\em tight example}, i.e., an infinite family of graphs on which the imputation computed is in the ${2/3}$-approximate core.

\begin{example}
	\label{ex.tight}
	Consider the following infinite family of graphs. For each $n$, the graph $G_n$ has $6n$ vertices $i_l, j_l, k_l$, for $1 \leq l \leq 2n$, and $6n$ edges $(i_l, j_l), (j_l, k_l), (i_l, k_l)$, for $1 \leq l \leq 2n$ all of weight 1. Clearly, $\opt(G_n) = 2n$ and $\opt_f(G_n) = 3n$. In case a connected graph is desired, add a clique on the $2n$ vertices $i_l$, for $1 \leq l \leq 2n$, with the weight of each edge being $\epsilon$, where $\epsilon$ tends to zero.
\end{example}

Observe that for the purpose of Lemma \ref{lem.app-cover}, we could have defined $f$ simply as $\forall i \in V, \ f(i) = {2 \over 3}$. However in general, this would have left a good fraction of the worth of the game unallocated. The definition of $f$ given above improves the allocation for agents who are in large odd cycles and those who are not in odd cycles with respect to matching $x$. As a result, the gain of a typical sub-coalition on seceding will be less than a factor of ${3 \over 2}$, giving it less incentive to secede. One way of formally stating an  improved factor is given in Proposition \ref{prop.improvement}; its proof is obvious from that of Theorem \ref{thm.main}. 

\begin{proposition}
	\label{prop.improvement}
	Assume that the underlying graph $G$ has no odd cycles of length less than $2k+1$. Then  imputation $c$ is in the ${{2k} \over {2k+1}}$-approximate core of the matching game for $G$.
	\end{proposition}

Finally, we show in Theorem \ref{thm.tight} that the integrality gap of LP-relaxation (\ref{eq.core-primal}) is precisely ${2 \over 3}$. As a consequence of this fact, improving the approximation factor of an imputation for the matching game is not possible.

\begin{theorem}
	\label{thm.tight}
		The integrality gap of LP-relaxation (\ref{eq.core-primal}) is ${2 \over 3}$.
\end{theorem}

\begin{proof}
	From the proof of the first part of Theorem \ref{thm.main} we get:
	\[ w(T) \ \geq \ \sum_{i \in V} {c_i} \ \geq \ {2 \over 3} \cdot \sum_{i \in V} {v_i} \ = \ {2 \over 3} \cdot w(x)  .\]
	Therefore for any instance $I = (G, w)$, 
		\[  {{\opt(I)} \over {\opt_f(I)}} \geq {2 \over 3} . \] 
This places a lower bound of ${2 \over 3}$ the integrality gap of LP-relaxation (\ref{eq.core-primal}). 

An upper bound of ${2 \over 3}$ on the integrality gap of LP-relaxation (\ref{eq.core-primal}) is placed by the example given in Example \ref{ex.tight}. 
\end{proof}


\section{How Discerning are Core Imputations?}
\label{sec.Degen}

In Theorem \ref{thm.degen} we provide new insights on core imputations of an assignment game by studying them via the lens of complementary slackness. We present a relationship between the competitiveness of individuals and teams of agents and the amount of profit they accrue in imputations that lie in the core, where by {\em competitiveness} we mean whether an individual or a team is matched in every/some/no maximum matching. Theorem \ref{thm.degen} shows once again how discerning core imputations are. Additionally, it sheds light on the phenomenon of degeneracy in assignment games, i.e., when the maximum weight matching is not unique. 

\begin{definition}
	\label{def.team}
	By a {\em mixed doubles team} we mean an edge in $G$; a generic one will be denoted as $e = (u, v)$. We will say that $e$ is {\em viable} if there is a maximum weight matching $M$ in $G$ such that $e \in M$ and it is {\em subpar} if for every maximum weight matching $M$ in $G$, $e \notin M$. Let $y$ be an imputation in the core of the game. We will say that $e$ is {\em fairly paid in $y$} if $y_u + y_v = w_e$ and it is {\em overpaid} if $y_u + y_v > w_e$\footnote{Observe that by the first constraint of the dual LP (\ref{eq.core-dual-bipartite}), these are the only possibilities.}. Finally, we will say that $e$ is {\em always paid fairly} if it is fairly paid in every imputation in the core.
\end{definition}

\begin{definition}
	\label{def.player}
	A generic player in $U \cup V$ will be denoted by $v$. Let $y$ be an imputation in the core. We will say that $v$ {\em gets paid in $y$} if $y_v > 0$ and {\em does not get paid} otherwise. Furthermore, $v$ is {\em paid sometimes} if there is at least one imputation in the core under which $v$ gets paid, and it is {\em never paid} if it is not paid under every imputation. We will say that $v$ is {\em essential} if $v$ is matched in every maximum weight matching in $G$ and it is {\em not essential} otherwise. 
\end{definition}

\begin{remark}
	Clearly, if an assignment game is non-degenerate, then every team and every player is either always matched or always unmatched. 
\end{remark}

\begin{theorem}
	\label{thm.degen}
	The following hold:
	\begin{enumerate}
		\item For every team $e \in E$: 
		\[ e \ \mbox{is viable} \ \iff \ e \ \mbox{is always paid fairly} \]
		\item For every player $v \in V$: 
		\[ v \ \mbox{is essential} \ \iff \ v \ \mbox{is paid sometimes} \]  
	\end{enumerate}
\end{theorem}
	
\begin{proof}
The proofs follow by applying complementary slackness conditions and strict complementarity to the primal LP (\ref{eq.core-primal-bipartite}) and dual LP (\ref{eq.core-dual-bipartite}); see \cite{Sch-book} for formal statements of these facts. We will use Theorem \ref{thm.SS} stating that the set of imputations in the core of the game is precisely the set of optimal solutions to the dual LP.

{\bf 1).} Let $x$ and $y$ be optimal solutions to LP (\ref{eq.core-primal-bipartite}) and LP (\ref{eq.core-dual-bipartite}), respectively. By the Complementary Slackness Theorem, for each $e = (u, v) \in E: \ \ x_e \cdot (y_u + y_v - w_e) = 0$.

Suppose $e$ is viable. Then there is an optimal solution to the primal, say $x$, under which it is matched, i.e., $x_e > 0$. Let $y$ be an arbitrary optimal dual solution. Then, by the Complementary Slackness Theorem, $y_u + y_v = w_e$. Varying $y$ over all optimal dual solutions, we get that $e$ is always paid fairly. This proves the forward direction.

For the reverse direction, we will use strict complementarity. It implies that corresponding to each team $e$, there is a pair of optimal primal and dual solutions $x, y$ such that either $x_e = 0$ or $y_u + y_v = w_e$ but not both. 

For team $e$, assume that the right hand side of the first statement holds and that $x, y$ is a pair of optimal solutions for which strict complementarity holds for $e$. Since $y_u + y_v = w_e$ it must be the case that $x_e > 0$. Now, since the polytope defined by the constraints of the primal LP (\ref{eq.core-primal-bipartite}) has integral optimal vertices, there is a maximum weight matching under which $e$ is matched. Therefore $e$ is viable and the left hand side of the first statement holds. 

{\bf 2).} The proof is along the same lines and will be stated more succinctly. Again, let $x$ and $y$ be optimal solutions to LP (\ref{eq.core-primal-bipartite}) and LP (\ref{eq.core-dual-bipartite}), respectively. By the Complementary Slackness Theorem, for each $v \in V: \ y_v \cdot (x(\delta(v)) - 1) = 0$. 

Suppose $v$ is paid sometimes. Then, there is an imputation in the core, say $y$, such that $y_v > 0$. Therefore, for every primal optimal solution $x$, $x(\delta(v)) = 1$ and in every maximum weight matching in $G$, $v$ is matched. Hence $v$ is essential, proving the reverse direction. 

 Strict complementarity implies that corresponding to each player $v$, there is a pair of optimal primal and dual solutions $x, y$ such that either $y_v = 0$ or $x(\delta(v)) = 1$ but not both. Since we have already established that the second condition must be holding for $x$, we get that $y_v > 0$ and hence $v$ is paid sometimes.
\end{proof}

Negating both sides of the first statement proved in Theorem \ref{thm.degen} we get the following double-implication. For every team $e \in E$: 
		\[ e \ \mbox{is subpar} \ \iff \ e \ \mbox{is sometimes overpaid} \]

Clearly, this statement is equivalent to the first statement of Theorem \ref{thm.degen} and hence contains no new information. However, it may provide a new viewpoint. These two equivalent  statements yield the following assertion, which at first sight seems incongruous with what we desire from the notion of the core and the just manner in which it allocates profits:

\begin{center}
{\em Whereas viable teams are always paid fairly, subpar teams are sometimes overpaid.}
\end{center}

How can the core favor subpar teams over viable teams? However, there is a simple justification for this assertion: The players $u$ and $v$ of a subpar team $e = (u, v)$ generate profit by teaming up with other players and their total profit can indeed exceed $w_e$.

Similarly, the second statement of Theorem \ref{thm.degen} is equivalent to the following. For every player $v \in V$: 
		\[ v \ \mbox{is never paid} \ \iff \ v \ \mbox{is not essential} \]

\section{Discussion}
\label{sec.discussion}

In the $2/3$-approximate core imputation, observe that in an odd cycle of length $2k + 1$, $k$ pairs of agents are matched and one agent is left unmatched. As a consequence, monetary transfers may be needed from all $2k$ matched agents to the unmatched agent. What happens if monetary transfers to an agent are allowed from only a limited number of other agents? If so, what is the best approximation factor possible? See also Remark \ref{rem.TU}.

For the assignment game, Shapley and Shubik are able to characterize ``antipodal'' points in the core, i.e., imputations which are maximally distant. An analogous understanding of the ${2 \over 3}$-approximate core of the general graph matching game will be desirable. 

\section{Acknowledgements}
\label{sec.ack}

I wish to thank Federico Echenique and Thorben Trobst for valuable discussions.

	\bibliographystyle{alpha}
	\bibliography{refs}

\end{document}